\documentclass[journal, 11 pt]{IEEEtran}

\IEEEoverridecommandlockouts                              

\renewcommand{\baselinestretch}{2}
\usepackage{graphicx}
 \usepackage{epstopdf}
\usepackage{amssymb}
\usepackage{color}
\usepackage{floatrow}
\usepackage{float}
\newfloatcommand{capbtabbox}{table}[][\FBwidth]
\usepackage{float}

\newcommand{\Rmnum}[1]{\expandafter\@slowromancap\romannumeral #1@}
\makeatother
\usepackage{algorithmic}
\usepackage{algorithm}
\usepackage{mathrsfs,amssymb}
\usepackage{mathrsfs}
\usepackage{bm}
\usepackage{amssymb}
\usepackage{mathrsfs}
\usepackage{subfigure}
\usepackage{mathrsfs}
\usepackage{subfigure}
\usepackage{multirow}
\usepackage{longtable,booktabs}
\usepackage[normal]{caption}
\usepackage[justification=centering]{caption}
\usepackage{cite}
\usepackage{threeparttable}
\usepackage{booktabs}
\usepackage{graphicx}      
\usepackage{graphics} 
\usepackage{mathptmx} 
\usepackage{amsmath} 
\usepackage{amssymb}  
\usepackage{bm}
\usepackage{cases}
\usepackage{color}
\usepackage{caption}

\newtheorem{theorem}{Theorem}
\newtheorem{proposition}{Proposition}

\newtheorem{corollary}{Corollary}
\newtheorem{definition}{Definition}

\newtheorem{example}{Example}

\newtheorem{discussion}{Discussion}

\newenvironment{proof}{\textit{Proof}:}{\hfill $\blacksquare$\par}
\definecolor{darkgreen}{cmyk}{0.91,0,0.88,0.32}
\definecolor{orange}{cmyk}{0.0,.5,1,0.1}
\definecolor{purple}{cmyk}{0.2,1,0,0.1}
\definecolor{pink}{cmyk}{0,0.4,0.4,0}
\definecolor{gray}{cmyk}{0.2,0.2,0.2,0}
\definecolor{greener}{cmyk}{0.91,0,0.88,0.2}

\title{\LARGE \bf
Non-Blockingness Verification of Bounded Petri Nets Using Basis Reachability Graphs---An Extended Version With Benchmarks
}

\author{Chao~Gu,~\IEEEmembership{Graduate Student Member,~IEEE}
        ~Ziyue Ma,~\IEEEmembership{Member,~IEEE}
        ~Zhiwu Li,~\IEEEmembership{Fellow,~IEEE}
        and~\\Alessandro Giua,~\IEEEmembership{Fellow,~IEEE}
\thanks{C. Gu is with the School of Electro-Mechanical Engineering, Xidian University, Xi'an 710071, China, and also with DIEE, University of Cagliari, Cagliari 09124, Italy
        {\tt\small cgu1992@stu.xidian.edu.cn}}%
\thanks{Z. Ma is with the School of Electro-Mechanical Engineering, Xidian University, Xi'an 710071, China
        {\tt\small maziyue@xidian.edu.cn}}
\thanks{Z. Li is with the School of Electro-Mechanical Engineering, Xidian University, Xi'an 710071, China, and also with the Institute of Systems Engineering, Macau University of Science and Technology, Macau {\tt\small zhwli@xidian.edu.cn}}
\thanks{A. Giua is with DIEE, University of Cagliari, Cagliari 09124, Italy {\tt\small giua@unica.it}}
}

\begin{document}

\IEEEspecialpapernotice{{\rm This article is an extended version of the paper ``C. Gu, Z. Ma, Z. Li and A. Giua. Non-blockingness verification of bounded Petri nets using basis reachability graphs. \textit{IEEE Control Systems Letters}, doi: 10.1109/LCSYS.2021.3087937, 2021'' with benchmarks.}}
%
%
\onecolumn

\maketitle
\thispagestyle{empty}
\pagestyle{empty}

\newpage
\begin{abstract}

In this paper, we study the problem of non-blockingness verification by tapping into the basis reachability graph (BRG).
Non-blockingness is a property that ensures that all pre-specified tasks can be completed, which is a mandatory requirement during the system design stage.
We develop a condition of transition partition of a given net such that the corresponding \emph{conflict-increase BRG} contains sufficient information on verifying non-blockingness of its corresponding Petri net.
Thanks to the compactness of the BRG, our approach possesses practical efficiency since the exhaustive enumeration of the state space can be avoided.
In particular, our method does not require that the net is deadlock-free.

\end{abstract}

\newpage
\section{Introduction}
Discrete event systems (DESs) \cite{SL, lafortune2019discrete} are \textit{event-driven} systems whose state space can be described as a discrete set.
As a mathematical characterization for studying, modelling, and analyzing DES, Petri nets\cite{cabasino2016marking, lefebvre2017near} offer various vantages over automata.
For instance, states in Petri nets can be represented as vectors, namely \textit{marking}s; ergo, techniques such as linear algebra\cite{basile2018algebraic} can be applied.
On the other hand, structural-based approaches can be adopted to avert exhaustively enumerating the state space, therefore mitigating the \textit{state explosion} problem.

In DESs, \textit{non-blockingness}\cite{yin2018synthesis} is a property that ensures that all pre-specified tasks can be completed, which is a mandatory requirement during the system design stage.
Given its importance, efficient techniques are desired to verify if a given system is non-blocking.
Past works \cite{c1, uzam2002optimal, ghaffari2003design} propose several methods to verify the non-blockingness of a given Petri net: these methods are based on the reachability graph (RG) and \emph{theory of regions}.
On the other hand, \cite{zhao2013iterative, hu2015maximally} propose methods to synthesize non-blocking enforcing supervisors in some subclasses of Petri nets.

Recently, a \textit{semi-structural} analysis technique in Petri nets, called the \emph{basis reachability analysis}, was proposed \cite{c8}.
In the basis reachability analysis, only a subset of the reachable markings called \emph{basis markings} is enumerated and an automaton-like structure called \emph{basis reachability graph (BRG)} is constructed.
Initially created for \emph{diagnosis} problems \cite{c8}, basis reachability analysis has been gradually developed and adopted on solving other issues such as \emph{marking reachability} and \emph{opacity} problems, etc.

Although the BRG-based techniques have been proved to be operative and efficient, it is showed that a conventional BRG is, in general, not applicable to tackle the non-blockingness verification problem \cite{gu2020verification}.
The reason is that in a BRG there may exist some livelocks among a set of non-basis markings.
In such a case, the blocking behavior of the plant net cannot be detected by inspecting the structure of the BRG.
To overcome such a problem, in \cite{gu2020verification}, an augmented version of BRGs called \textit{minimax-BRG} is developed.
Although the minimax-BRG exhibits practical efficiency in solving the non-blockingness verification problem, unlike the basis-marking-based approach\cite{ru2008supervisor}, currently, there are no analysis methods for addressing with minimax-BRGs on state estimation and supervisory control problems, in which non-blockingness analysis plays a key role.
This motivated us to develop an alternative non-blocking verification method based on the conventional types of BRGs whose analysis methods are relatively mature.

In this paper we introduce a particular type of BRGs called \textit{conflict-increase BRGs} (CI-BRGs) that encode sufficient non-blockingness-related information by referring to a particular partition of the transition set.
Differently from the minimax-BRG\cite{gu2020verification} whose construction has a higher complexity than that of a BRG with the same transition partition, a CI-BRG is identical in essence with BRG.
We characterize the main properties of basis markings in CI-BRGs and prove that the non-blockingness of a system can be verified by checking if all basis markings in its corresponding CI-BRG are non-blocking.
Although there exist restrictions on obtaining of CI-BRGs, which depend on the system structure and the parameters of the linear constraint that describes the final markings set, thanks to the compactness of BRGs, our approach still achieves practical efficiency compared with the RG-based analysis, according to numerical results.


\section{Preliminaries}
%

\subsection{Petri nets}
A Petri net is a four-tuple $N=(P,T,Pre,Post)$, where $P$ is a set of $m$ \textit{places} and $T$ is a set of $n$ \textit{transitions}.
$Pre: P\times T\rightarrow \mathbb{N}$ and $Post: P\times T\rightarrow \mathbb{N}$ ($\mathbb{N}=\{0, 1, 2, \cdots\}$) are the \textit{pre}- and \textit{post}- \textit{incidence functions} that specify the \textit{arcs} in the net and are represented as matrices in $\mathbb{N}^{m\times n}$.
The \emph{incidence matrix} of $N$ is defined by $C=Post-Pre$.
A Petri net is \textit{acyclic} if there are no directed cycles in its underlying digraph.

Given a Petri net $N=(P,T,Pre,Post)$ and a set of transitions $T_x\subseteq T$, the \textit{$T_x$-induced sub-net} of $N$ is a net resulting by removing all transitions in $T\setminus T_x$ and corresponding arcs from $N$, denoted as $N_x=(P,T_x,Pre_x,Post_x)$ where $T_x\subseteq T$ and $Pre_x$ ($Post_x$) is the restriction of $Pre$ ($Post$) to $P$ and $T_x$.
The incidence matrix of $N_x$ is denoted by $C_x = Post_x-Pre_x$.

A \emph{marking} $M$ of a Petri net $N$ is a mapping: $P\to\mathbb{N}$ that assigns to each place of a Petri net a non-negative integer number of \textit{tokens}. The number of tokens in a place $p$ at a marking $M$ is denote by $M(p)$.
A Petri net $N$ with an initial marking $M_0$ is called a \textit{marked net}, denoted by $\langle N, M_0\rangle$.
For a place $p\in P$, the \textit{set of its input transitions} is defined by $^{\bullet}p=\{t\in T\mid Post(p,t)>0\}$ and the \textit{set of its output transitions} is defined by $p^{\bullet}=\{t\in T\mid Pre(p,t)>0\}$. The notions for $^{\bullet}t$ and $t^{\bullet}$ are analogously defined.
A Petri net $N=(P, T, Pre, Post)$ is \textit{conflict-free} if for all $p\in P$, $|p^{\bullet}|\leq 1$.

A transition $t\in T$ is \emph{enabled} at a marking $M$ if $M\geq Pre(\cdot, t)$, denoted by $M[t\rangle$; otherwise it is said to be \textit{disabled} at $M$, denoted as $\neg M[t\rangle$.
If $t$ is enabled at $M$, the \emph{firing} of $t$ yields marking $M^{\prime}=M+C(\cdot, t)$, which is denoted as $M[t\rangle M^{\prime}$.
A marking $M$ is \textit{dead} if for all $t\in T$, $M\ngeqslant Pre(\cdot, t)$.

Marking $M^{\prime}$ is \emph{reachable} from $M_{1}$ if there exist a firing sequence of transitions $\sigma=t_{1}t_{2}\cdots t_{n}$ and markings $M_{2},\cdots, M_{n}$ such that $M_{1}[t_{1}\rangle M_{2}[t_{2}\rangle\cdots M_{n}[t_{n}\rangle M^{\prime}$ holds.
We denote by $T^*$ the set of all finite sequences of transitions over $T$.
Given a transition sequence $\sigma\in T^{*}$, $\varphi: T^{*}\rightarrow \mathbb{N}^{n}$ is a function that associates to $\sigma$ a vector $\textbf{y}=\varphi(\sigma)\in \mathbb{N}^{n}$, called the \textit{firing vector} of $\sigma$.
Let $\varphi^{-1}: \mathbb{N}^{n}\rightarrow T^{*}$ be the inverse function of $\varphi$, namely for $\textbf{y}\in \mathbb{N}^{n}$, $\varphi^{-1}(\textbf{y}):=\{\sigma\in T^{*}| \varphi(\sigma)=\textbf{y}\}$.
The set of markings reachable from $M_{0}$ is called the \emph{reachability set} of $\langle N, M_0\rangle$, denoted by $R(N, M_{0})$.
A marked net $\langle N, M_0\rangle$ is said to be \emph{bounded} if there exists an integer $k\in \mathbb{N}$ such that for all $M\in R(N, M_0)$ and for all $p\in P$, $M(p)\leq k$ holds; otherwise it is said to be \textit{unbounded}.

\begin{proposition}$\!\!\!\!${\rm\cite{Murata}}\label{ProX}
Given a marked net $\langle N, M_0\rangle$ where $N$ is acyclic, $M\in R(N, M_0)$, $M^{\prime}\in R(N, M_0)$ and a firing vector $\textbf{y}\in \mathbb{N}^n$, the following holds:
\begin{center}
$\ \ \ \ \ \ M^{\prime}=M+C\cdot \textbf{y}\geq \textbf{0}\Leftrightarrow (\exists \sigma\in \varphi^{-1}(\textbf{y}))\ M[\sigma\rangle M^{\prime}.\hfill\square$
\end{center}
\end{proposition}

Let $G=(N, M_0, {\cal F})$ denote a \textit{plant} consisting of a marked net and a finite set of final markings ${\cal F}\subseteq R(N, M_0)$.
Instead of explicitly listing all the elements in ${\cal F}$, as a general form, in this paper, we characterize set ${\cal F}$ as a linear constraints namely \textit{generalized mutual exclusion constraints} (GMECs)\cite{c25}.
A GMEC is a pair $(\textbf{w}, k)$, where $\textbf{w}\in \mathbb{Z}^m$ and $k\in \mathbb{Z}$ ($\mathbb{Z}$ is the set of integers), that defines a set of markings $\mathcal{L}_{(\textbf{w},k)}=\{M\in\mathbb{N}^m| \textbf{w}^T\cdot M \leq k\}.$


\begin{definition}\label{NB}
{\rm A marking $M\in R(N, M_0)$ of a plant $G=(N, M_0, {\cal F})$ is said to be \emph{blocking} if $R(N, M)\cap {\cal F} =\emptyset$; otherwise $M$ is said to be \textit{non-blocking}.
Plant $G$ is \textit{non-blocking} if all reachable markings are non-blocking; otherwise $G$ is \textit{blocking}.$\hfill\diamondsuit$}
\end{definition}
%

\subsection{Basis Marking and Basis Reachability Graph}\label{SectionIIB}

Given a Petri net $N = (P, T, Pre, Post)$, transition set $T$ can be partitioned into $T=T_E\cup T_I$ where the sets $T_E$ and $T_I$ are called the \textit{explicit} transition set and the \textit{implicit} transition set, respectively.
A pair $\pi=(T_E, T_I)$ is called a \textit{basis partition}\cite{c3} of $T$ if the $T_I$-induced subnet is acyclic. We denote $|T_E|=n_E$ and $|T_I|=n_I$.
Let $C_I$ be the incidence matrix of the $T_I$-induced subnet.
Note that in a BRG with basis partition $(T_E, T_I)$, the firing of explicit transitions in $T_E$ is explicitly represented in the BRG, while the firing of implicit transitions in $T_I$ is abstracted.
Note that no physical meaning needs to be associated with implicit transitions: the set $T_I$ can be arbitrarily selected, provided that the $T_I$-induced subnet is acyclic.

\begin{definition}
Given a Petri net $N = (P, T, {\rm Pre}, {\rm Post})$, a basis partition $\pi=(T_E, T_I)$, a marking $M$, and a transition $t\in T_E$, we define
\begin{center}
$\Sigma(M, t)=\{\sigma\in T_{I}^{\ast}| M[\sigma\rangle M^{\prime}, M^{\prime}\geq {\rm Pre} (\cdot, t)\}$
\end{center}
as the set of \emph{explanations} of $t$ at $M$, and
\begin{center}
$Y(M, t)=\{\varphi(\sigma)\in \mathbb{N}^{n_I}| \sigma\in \Sigma(M, t)\}$
\end{center}
as the set of \emph{explanation vectors}; meanwhile we define
\begin{center}
$\Sigma_{{\rm min}}(M, t)=\{\sigma\in \Sigma(M, t)| \nexists \sigma^{\prime}\in \Sigma(M, t): \varphi(\sigma^{\prime})\lneq \varphi(\sigma)\}$
\end{center}
as the set of \emph{minimal explanations} of $t$ at $M$, and

\begin{center}
$Y_{{\rm min}}(M, t)=\{\varphi(\sigma)\in \mathbb{N}^{n_{I}}| \sigma\in \Sigma_{{\rm min}}(M, t)\}$
\end{center}
as the corresponding set of \emph{minimal explanation vectors}.$\hfill\diamondsuit$
\end{definition}


\begin{definition}
Given a bounded marked net $\langle N, M_0\rangle$ with a basis partition $\pi=(T_E, T_I)$, its \emph{basis reachability graph} (BRG) is a deterministic automaton $\mathcal{B}$ output by Algorithm 2 in \cite{c3}. The BRG $\mathcal{B}$ is a quadruple $(\mathcal{M_B}, {\rm Tr}, \Delta, M_0)$, where the state set $\mathcal{M_B}$ is the set of basis markings, the event set ${\rm Tr}$ is the set of pairs $(t, y)\in T_E\times \mathbb{N}^{n_{I}}$, the transition relation $\Delta=\{(M_1, (t, y), M_2)| t\in T_E, y\in Y_{\rm min}(M_1, t), M_2=M_1+C_I\cdot y+C(\cdot, t)\}$, and the initial state is the initial marking $M_0$.$\hfill\diamondsuit$
\end{definition}

The set $Y_{{\rm min}}(M, t)$ and the BRG can be constructed through Algorithms 1 and 2 in \cite{c3}, respectively.
We extend the definition of transition relation to consider sequence of pairs $\sigma\in {\rm Tr}^*$ and write $M_1\xrightarrow{\sigma}M_2$ to denote that from $M_1$ sequence $\sigma$ yields $M_2$.

Note that the upper bound of states in a BRG is the size of the reachability space of a net. However, many BRG-related work \cite{c8,c3} have shown that in practical cases a BRG can be much smaller than the corresponding reachability space, i.e., $|\mathcal{M_{B}}|\ll|R(N, M_0)|$ holds.
Besides, in some cases, a BRG may grow much slower than the reachability space of a net.
For instance, \cite{tong2016verification} shows an example in which the number of states in a Petri net grows cubically (i.e., $O(k^3)$) when the initial marking increases, while the corresponding BRG growth linearly (i.e., $O(k)$).
Therefore, the construction of the BRG achieves practical efficiency.

\begin{definition}
Given a marked net $\langle N, M_0\rangle$, a basis partition $\pi=(T_E, T_I)$, and a basis marking $M_b\in \mathcal{M_B}$, we define $R_I(M_b)$ the \emph{implicit reach} of $M_b$ as:

\begin{center}
$R_I(M_b)=\{M\in \mathbb{N}^m|(\exists \sigma\in T_I^\ast)\; M_b[\sigma\rangle M\}$.
\end{center}

Since the $T_I$-induced subnet is acyclic, we have:

\begin{center}
$R_I(M_b)=\{M\in\mathbb{N}^m| (y_I\in\mathbb{N}^{n_I})\; M=M_b+C_I\cdot y_I\}$.
\end{center}
\end{definition}

\begin{proposition}$\!\!\!\!${\rm\cite{c3}}\label{BRGpro}
Given a marked net $\langle N, M_0\rangle$ with a basis partition $\pi=(T_E, T_I)$, the set of basis markings of the system is $\mathcal{M_{B}}$.
Consider a marking $M\in\mathbb{N}^m$.
$M\in R(N, M_0)$ if and only if there exists a basis marking $M_b\in \mathcal{M_{B}}$ such that $M\in R_I(M_b)$.$\hfill\diamondsuit$
\end{proposition}

\section{Non-blockingness Verification Using Conflict-Increase BRGs}\label{Main}

Given a plant net, in general there exist several valid basis partitions, each of which leads to a different BRG.
According to Example 1 presented in \cite{Gu}, a BRG constructed with a randomly selected basis partition may not encode all information needed to test if a plant is non-blocking. The reason lies in the fact that in a BRG there may exist some \textit{livelocks} among a set of non-basis markings; thus, the blocking
behavior of the plant net cannot be detected by checking the structure of the BRG.

\subsection{Conflict-Increase BRGs}
In this part, we show how a BRG corresponding to a suitable basis partition may allow one to verify non-blockingness.
For a given plant, we first introduce the notion of \textit{conflict-increase BRGs} (CI-BRGs).

\begin{definition}\label{CP}
Consider a plant $G=(N, M_0, {\cal F})$ with $N=(P, T, Pre, Post)$ and $\cal{F} = \mathcal{L}_{(\textbf{w},{\textit k})}$.
A transition set $T'\subseteq T$ is said to be \emph{non-conflicting} if $T'\subseteq T\setminus T_{conf}$, where
\begin{equation}\label{eq:conf}
T_{conf}=\{t\in T\mid (\exists p\in P) t\in p^\bullet, |p^\bullet|\geq 2\};
\end{equation}
it is said to be \emph{non-increasing} if $T'\subseteq T\setminus T_{inc}$, where
\begin{equation}\label{eq:pos}
T_{inc}=\{t\in T\mid \mathbf{w}^T\cdot C(\cdot, t)> 0\}.
\end{equation}$\hfill\diamondsuit$
\end{definition}
According to Definition \ref{CP}, a transition set is non-conflicting and non-increasing if it does not contain two types of transitions:
\begin{itemize}
  \item [(1)] $T_{conf}$ --- all transitions that are in structural conflicts (depends only on the structure of the net);
  \item [(2)] $T_{inc}$ --- all transitions whose influence on $(\mathbf{w}, k)$ are positive (depends only on the corresponding GMEC), i.e., at a marking $M$, by firing a transition $t\in T_{inc}$, the token count of the obtained marking $M^{\prime} = M+C(\cdot, t)$ will be increased in terms of the corresponding GMEC since $\mathbf{w}^T\cdot C(\cdot, t)> 0$.
\end{itemize}

\begin{definition}\label{TBasis}
Consider a plant $G=(N, M_0, {\cal F})$ with $N=(P, T, Pre, Post)$ and $\cal{F} = \mathcal{L}_{(\textbf{w},{\textit k})}$.
A BRG of $G$ with respect to $\pi = (T_E, T_I)$ is called a \textit{conflict-increase BRG} (CI-BRG) if $T_I$ is non-conflicting and non-increasing.$\hfill\diamondsuit$
\end{definition}


Note that to construct a CI-BRG, the corresponding implicit transition set $T_I$ in $\pi = (T_E, T_I)$ should be
both non-conflicting and non-increasing.
In addition to that, some other transitions (if necessary) may also need to be eliminated in $T_I$ to ensure the acyclicity of the $T_I$-induced subnet of $G$.
Moreover, for a bounded plant, there always exists a CI-BRG (e.g., the BRG with respect to $T_E = T$ and $T_I = \emptyset$).
Since CI-BRG is a particular type of BRG, it can be computed by Algorithm 2 in \cite{c3} and the complexity analysis of the BRG (described in Section \ref{SectionIIB}) also applies to the CI-BRG.

  On the other hand, in the previous works \cite{c3,c8}, it is often preferred to choose a basis partition $\pi$ with a $T_I$ that is maximal (in the sense of set containment) since the size of its corresponding BRG is relatively small.
  Here, a non-conflicting and non-increasing $T_I$ may not be maximal, therefore the corresponding CI-BRG may not be minimal as well.
  However, as a trade-off, we show hereinafter that necessary information regarding non-blockingness will be appropriately encoded in CI-BRGs and the non-blockingness verification procedure can be therefore facilitated.
  Moreover, as shown by simulations in Section \ref{Section5}, a CI-BRG is still significantly smaller than the corresponding reachability graph in size.


\begin{example}\label{Example22}
  \begin{figure}[t]
\includegraphics[width=6cm]{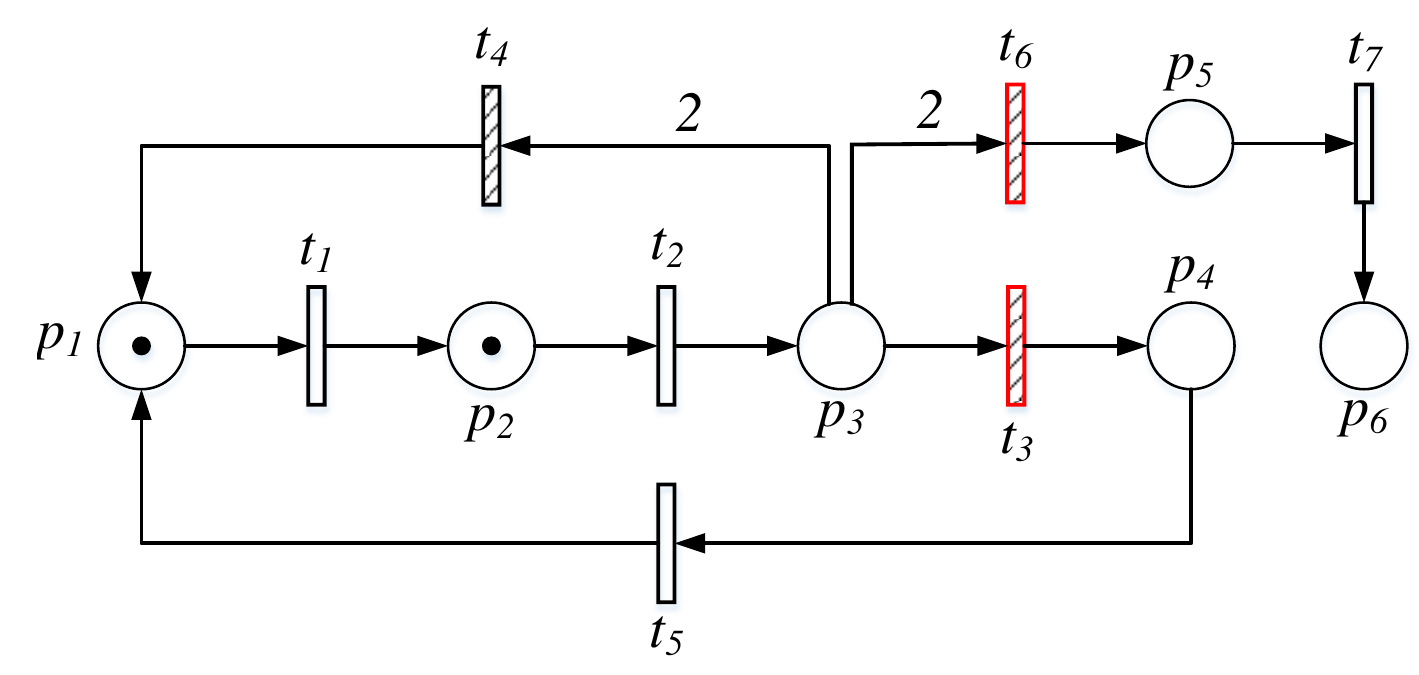}
\caption{A plant $G$ with $\pi = (T_E, T_I)$ where $T_E = \{t_3, t_4, t_6\}$.}\label{plant2}
\end{figure}
\begin{figure}[t]
\includegraphics[width=8cm]{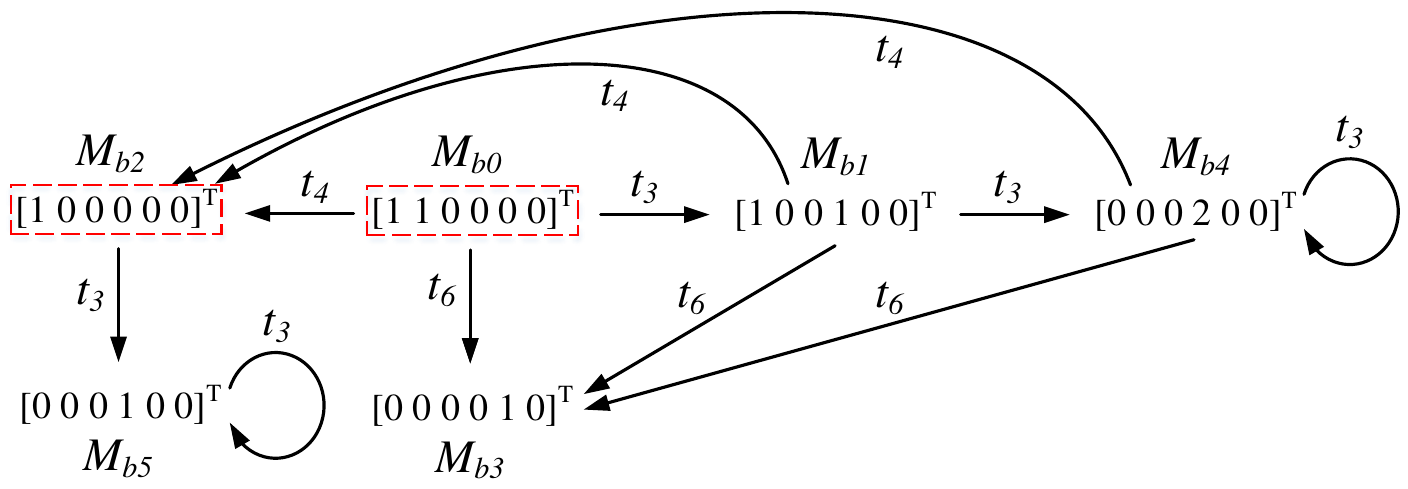}
\caption{The CI-BRG $\cal{B}$. For readability, all explanation vectors are not shown.}\label{CPBRG1}
\end{figure}
Consider a plant $G = (N, M_0, \cal{F})$ in Fig. \ref{plant2} with $M_0 = [1\ 1\ 0\ 0\ 0\ 0]^{T}$ and $\cal{F} = \mathcal{L}_{(\textbf{w},{\textit k})}$ where $\textbf{w} = [0\ 0\ 0\ 1\ 1\ 1]^{T}$ and $k = 0$.
  In this net, $T_{conf} = \{t_3, t_4, t_6\}$ (marked with shadow) and $T_{inc} = \{t_3, t_6\}$ (boxed in red); thus, $T_{conf}\cup T_{inc}=\{t_3, t_4, t_6\}$.
  Notice that the subnet induced by $T\setminus\{t_3, t_4, t_6\}=\{t_1, t_2, t_5, t_7\}$ is acyclic.
  Therefore, we can choose the set of implicit transitions $T_I = \{t_1, t_2, t_5, t_7\}$ that is non-conflicting and non-increasing, which leads to $T_E = \{t_3, t_4, t_6\}$.

  Now we construct the CI-BRG $\mathcal{B} = (\mathcal{M_B}, {\rm Tr}, \Delta, M_0)$.
  According to Algorithm 2 in \cite{c3}, as the initialization, let the basis marking set $\mathcal{M_{B}} = \{M_{b0}\}$ where $M_{b0} = M_0$, the set of pairs ${\rm Tr} = \emptyset$, and the transition relation set $\Delta = \emptyset$.
  Then, at $M_{b0}$, we compute all minimal explanation vectors $\textbf{y}_{min}$ of all explicit transitions (i.e., $t_3, t_4$, and $t_6$) at $M_{b0}$ through Algorithm 1 in \cite{c3} and derive the corresponding basis markings $M_b$:
  \begin{itemize}
    \item for $t_3$, $\textbf{y}_{min} = [0\ 1\ 0\ 0]^T$. Since $M_{b0}+C_I\cdot \textbf{y}_{min} + C (\cdot, t_3) = [ 1\ 0\ 0\ 1\ 0\ 0]^T\notin \mathcal{M_{B}}$, let $M_{b1} = [ 1\ 0\ 0\ 1\ 0\ 0]^T$ and $\mathcal{M_{B}} = \{M_{b0}, M_{b1}\}$. Meanwhile, let ${\rm Tr} = \{(t_3, [0\ 1\ 0\ 0]^T)\}$ and $\Delta = \{(M_{b0}, (t_3, [0\ 1\ 0\ 0]^T), M_{b1})\}$;
    \item for $t_4$, $\textbf{y}^{\prime}_{min} = [1\ 2\ 0\ 0]^T$. Since $M_{b0}+C_I\cdot \textbf{y}^{\prime}_{min} + C (\cdot, t_4) = [ 1\ 0\ 0\ 0\ 0\ 0]^T\notin \mathcal{M_{B}}$, let $M_{b2} = [ 1\ 0\ 0\ 0\ 0\ 0]^T$ and $\mathcal{M_{B}} = \{M_{b0}, M_{b1}, M_{b2}\}$. Meanwhile, let ${\rm Tr} = \{(t_3, [0\ 1\ 0\ 0]^T), (t_4, [1\ 2\ 0\ 0]^T)\}$ and $\Delta = \\ \{(M_{b0}, (t_3, [0\ 1\ 0\ 0]^T), M_{b1}), (M_{b0}, (t_4, [1\ 2\ 0\ 0]^T), M_{b2})\}$;
    \item for $t_6$, $\textbf{y}^{\prime\prime}_{min} = [1\ 2\ 0\ 0]^T$. Since $M_{b0}+C_I\cdot \textbf{y}^{\prime\prime}_{min} + C (\cdot, t_6) = [ 0\ 0\ 0\ 0\ 1\ 0]^T\notin \mathcal{M_{B}}$, let $M_{b3} = [ 0\ 0\ 0\ 0\ 1\ 0]^T$ and $\mathcal{M_{B}} = \{M_{b0}, M_{b1}, M_{b2}, M_{b3}\}$. Meanwhile, let ${\rm Tr} = \{(t_3, [0\ 1\ 0\ 0]^T), (t_4, [1\ 2\ 0\ 0]^T), (t_6, [1\ 2\ 0\ 0]^T)\}$ and $\Delta = \{(M_{b0}, (t_3, [0\ 1\ 0\ 0]^T), M_{b1}), (M_{b0}, (t_4, [1\ 2\ 0\ 0]^T), M_{b2}),\\ (M_{b0}, (t_6, [1\ 2\ 0\ 0]^T), M_{b3})\}$.
  \end{itemize}

  Analogously, we compute all minimal explanation vectors of all explicit transitions at $M_{b1}, M_{b2}$ and $M_{b3}$ and then derive new basis markings and transition relations. Iteratively, it can be computed that:
   \begin{itemize}
     \item $\mathcal{M_{B}} = \{M_{b0}, M_{b1}, M_{b2}, M_{b3}, M_{b4}, M_{b5}\}$ where $M_{b4} = [0\ 0\ 0\ 2\ 0\ 0]^T$ and $M_{b5} = [0\ 0\ 0\ 1\ 0\ 0]^T$;
     \item ${\rm Tr} = \{(t_3, [0\ 1\ 0\ 0]^T), (t_4, [1\ 2\ 0\ 0]^T), (t_6, [1\ 2\ 0\ 0]^T), \\(t_3, [1\ 1\ 0\ 0]^T), (t_3, [1\ 1\ 1\ 0]^T), (t_4, [2\ 2\ 1\ 0]^T),\\ (t_6, [2\ 2\ 1\ 0]^T), (t_4, [2\ 2\ 2\ 0]^T), (t_6, [2\ 2\ 2\ 0]^T)\}$;
     \item $\Delta = \{(M_{b0}, (t_3, [0\ 1\ 0\ 0]^T), M_{b1}),\\ (M_{b0}, (t_4, [1\ 2\ 0\ 0]^T), M_{b2}), (M_{b0}, (t_6, [1\ 2\ 0\ 0]^T), M_{b3}),\\ (M_{b1}, (t_3, [1\ 1\ 0\ 0]^T), M_{b4}), (M_{b1}, (t_4, [2\ 2\ 1\ 0]^T), M_{b2}),\\ (M_{b1}, (t_6, [2\ 2\ 1\ 0]^T), M_{b3}), (M_{b2}, (t_3, [1\ 1\ 0\ 0]^T), M_{b5}),\\ (M_{b4}, (t_3, [1\ 1\ 1\ 0]^T), M_{b4}), (M_{b4}, (t_4, [2\ 2\ 2\ 0]^T), M_{b2}),\\ (M_{b4}, (t_6, [2\ 2\ 2\ 0]^T), M_{b3}), (M_{b5}, (t_3, [1\ 1\ 1\ 0]^T), M_{b5})\}$.
   \end{itemize}
%

  In summary, the CI-BRG $\mathcal{B}= (\mathcal{M_B}, {\rm Tr}, \Delta, M_0)$ is graphically presented in Fig. \ref{CPBRG1}.
  Meanwhile, the basis markings that are also final (i.e., $M_{b0}$ and $M_{b2}$) in $\mathcal{B}$ are boxed with red dashed lines.$\hfill\diamondsuit$
 \end{example}

\subsection{Properties of CI-BRGs}
In this subsection we prove a series of results on the properties of CI-BRGs.
These results will be eventually used to establish our non-blocking verification algorithm.
Before entering the mathematical details, we note that here we do \emph{not} assume that the net is deadlock-free: the results presented in this subsection hold for both deadlock-free nets and nets with deadlocks.
The cases of non-deadlock-free nets will be further discussed in the next subsection.

First, the following proposition shows that for any basis marking $M_b$, if an explicit transition $t$ is enabled by firing a minimal explanation from $M_b$, then $t$ remains enabled when any other implicit transitions fires.

\begin{proposition}\label{PropOne}
Given a plant $G=(N, M_0, {\cal F})$ with $\cal{F} = \mathcal{L}_{(\textbf{w},{\textit k})}$, let $\cal{B}$ be its CI-BRG with respect to $\pi= (T_E, T_I)$.
Let $\mathbf{y}_{min}\in Y_{min}(M_b, t)$ be a minimal explanation vector of explicit transition $t$ at $M_b$. For any implicit firing vector $\mathbf{y}\geq\mathbf{y}_{min}$, $M_b+C_I\cdot\mathbf{y}=M\geq\mathbf{0}$ implies $M[t\rangle$.
\end{proposition}
\begin{proof}
According to Proposition~1, there exists $\sigma_{min}\in T_I^*$ such that $M_b[\sigma_{min}\rangle M_{min}\geq\mathbf{0}$ where $\varphi(\sigma_{min}) = \mathbf{y}_{min}$.
Since $\mathbf{y}_{min}\leq\mathbf{y}$, the following trajectory is feasible:
\begin{center}
  $M_b[\sigma_{min}\rangle M_{min}[\sigma'\rangle M,$
\end{center}
where $\varphi(\sigma^{\prime})=\mathbf{y}-\mathbf{y}_{min}.$
By Eq.~\eqref{eq:conf}, in $G$ no implicit transition shares any input place with transition $t$.
Hence, once $t$ is enabled at $M_{min}$, it remains enabled regardless other implicit transitions fire, which implies $M[t\rangle$.
\end{proof}

Next, we define the concept of \textit{maximal implicit firing vector} and \textit{i-maximal marking} as follows.
\begin{definition}\label{DefOne}
Given a plant $G=(N, M_0, {\cal F})$ with $\cal{F} = \mathcal{L}_{(\textbf{w},{\textit k})}$, let $\cal{B}$ be its CI-BRG with $\pi = (T_E, T_I)$.
At a basis marking $M_b$ in $\mathcal{B}$, an implicit firing vector (IFV) $\mathbf{y}\in \mathbb{N}^{n_I}$ is said to be \textit{maximal} if there exists a firing sequence $\sigma\in T_I^*, \sigma\in \varphi^{-1}(\mathbf{y})$ such that $M_b[\sigma\rangle$ and there does not exist any other $\sigma^{\prime}\in T_I^*, \varphi(\sigma)\gneq \mathbf{y}$ such that $M_b[\sigma\rangle M$.
A marking $M_{\max}$ is said to be \emph{i-maximal} at $M_b$ if there exists a maximal IFV $\mathbf{y}$ such that $M_b+ C_I\cdot\mathbf{y} =  M_{\max}$.
$\hfill\diamondsuit$
\end{definition}

Further, we show in Proposition \ref{PropTwo} that,  in a CI-BRG, for any basis marking there exists a unique maximal IFV and a unique i-maximal marking.


\begin{proposition}\label{PropTwo}
Given a marked net $\langle N, M_0\rangle$, let $\cal{B}$ be its BRG with respect to $\pi = (T_E, T_I)$ where the $T_I$-induced sub-net of $N$ is conflict-free.
At any basis marking $M_b$ in $\cal{B}$, there exists a unique maximal IFV $\mathbf{y}$.
\end{proposition}
\begin{proof}
By contradiction, suppose that at $M_b$ there exist two different maximal IFVs $\mathbf{y}_1, \mathbf{y}_2$.
Since the $T_I$-induced subnet is acyclic, there exist two firing sequences $\sigma_1=t_{i_1}^{k_1}\cdots t_{i_n}^{k_n}, \sigma_2=t_{i_1}^{k_1'}\cdots t_{i_n}^{k_n'}$ corresponding to $\mathbf{y}_1, \mathbf{y}_2$, respectively, in which $t_{i_j}$'s are sorted from upstream to downstream of the $T_I$-induced subnet.
Since $\mathbf{y}_1\neq\mathbf{y}_2$, there exists a minimal $j$ such that $k_j\neq k_j'$ and for all $j'<j$, $k_j= k_{j'}'$.
Without loss of generality, suppose that $k_j< k_j'$.
Consider $\bar\sigma_1=t_{i_1}^{k_1}\cdots t_{i_{j-1}}^{k_{j-1}} t_{i_j}^{k_j}$ that is a prefix of $\sigma_1$.
Since $t_{i_1}^{k_1}\cdots t_{i_{j-1}}^{k_{j-1}} t_{i_j}^{k_j'}$ is a prefix of $\sigma_2$, transition $t_j$ can fire at least once after $\bar\sigma_1$.
Moreover, since the transitions in $\sigma_1$ appears from upstream to downstream and the $T_I$-induced subnet is conflict-free, the firing of $t_j$ after $\bar\sigma_1$ does not affect the firing of the rest of transitions in $\sigma_1$, which implies that sequence $t_{i_1}^{k_1}\cdots t_{i_j}^{k_j+1}\cdots t_{i_n}^{k_n}$ is firable at $M_b$.
This means that $\mathbf{y}_1$ is not a maximal IFV.
\end{proof}

According to Proposition \ref{PropTwo}, it can be inferred that in a CI-BRG each basis marking has a unique i-maximal marking.
In the sequel, for each basis marking $M_b$ we denote its i-maximal marking as $M_{b, max}$.
From Propositions~\ref{PropOne} and~\ref{PropTwo} we immediately have the following Proposition.
\begin{proposition}\label{CoroOne}
Given a marked net $\langle N, M_0\rangle$ with $N=(P, T, Pre, Post)$, let $\cal{B}$ be its BRG with respect to $\pi = (T_E, T_I)$ where $T_I$ is non-conflicting.
For any basis marking $M_b$ in $\cal{B}$, it holds that
\begin{center}
  $M_b[\sigma_{min}t\rangle \Rightarrow M_{b,\max}[t\rangle,$
\end{center}
where $t\in T_E$, $\sigma_{min}\in \Sigma_{min}(M_b, t)$, and $M_{b,\max}$ is an i-maximal marking at $M_b$.
\end{proposition}
\begin{proof}
Since $T_I$ is non-conflicting, any transition $t\in T_E$ is not in conflict with transitions in $T_I$.
Also, the fact that $M_{b,\max}$ is an i-maximal marking at $M_b$ implies that there exists a maximal IFV $\mathbf{y}\in \mathbb{N}^{n_I}$ at $M_b$ such that $M_b+ C_I\cdot\mathbf{y} =  M_{b,\max}$ holds.
Thus, this statement follows Proposition \ref{PropOne}, since $\mathbf{y}\geq \varphi(\sigma_{min})$ by Definition \ref{DefOne}.
\end{proof}

Next, we show in Proposition \ref{PropThree} that in a CI-BRG, for any marking $M$ in the implicit reach of $M_b$, there necessarily exists a firing sequence consisting of implicit transitions $\sigma$ in $T_I^*$ such that $M[\sigma\rangle M_{b,\max}$.

\begin{proposition}\label{PropThree}
Given a plant $G=(N, M_0, {\cal F})$ with $N=(P, T, Pre, Post)$ and $\cal{F} = \mathcal{L}_{(\textbf{w},{\textit k})}$, let $\cal{B}$ be its CI-BRG with respect to $\pi = (T_E, T_I)$.
Given a basis marking $M_b$ and its i-maximal marking $M_{b,\max}\in R_I(M_b)$, the following holds:
\begin{center}
  $(\forall M\in R_I(M_b), \exists \sigma\in T_I^*)\ M[\sigma\rangle M_{b,\max}.$
\end{center}
\end{proposition}
\begin{proof}
Consider trajectories $M_b[\sigma_I\rangle M$ and $M_b[\sigma_{max}\rangle M_{b,\max}$, where $\varphi(\sigma_{max}) = \textbf{y}_{max}\in \mathbb{N}^{n_I}$ is the unique maximal IFV at $M_b$ and $\varphi(\sigma_I) = \textbf{y}_I\in \mathbb{N}^{n_I}$.
Since $M=M_b+C_I\cdot\mathbf{y}_I\geq\mathbf{0}$ and $M_{b,\max}=M_b+C_I\cdot\mathbf{y}_{max}\geq\mathbf{0}$, it holds that $M_{b,\max}=M+C_I\cdot (\mathbf{y}_{max}-\mathbf{y}_I)$. Since $\mathbf{y}_{max}\geq\mathbf{y}_I$ (according to Definition \ref{DefOne}), based on Proposition~1, there must exists a firing sequence $\sigma\in T_I^*$ such that $M[\sigma\rangle M_{b,\max}$ where $\varphi(\sigma) = \mathbf{y}_{max}-\mathbf{y}_I$.
\end{proof}

The next proposition shows that in a CI-BRG, the implicit reach of any basis marking $M_b$ contains at least one final marking (i.e., $R_I(M_b)\cap{\cal F}\neq\emptyset$) if and only if the i-maximal marking of $M_b$ is a final marking (i.e., $M_{b, \max}\in{\cal F}$).

\begin{proposition}\label{ProFour}
Given a plant $G=(N, M_0, {\cal F})$ with $N=(P, T, Pre, Post)$ and $\cal{F} = \mathcal{L}_{(\textbf{w},{\textit k})}$, let $\cal{B}$ be its CI-BRG with respect to $\pi = (T_E, T_I)$.
For any basis marking $M_b$ and its i-maximal marking $M_{b,\max}\in R_I(M_b)$, it holds:
\begin{center}
  $R_I(M_b)\cap{\cal F}\neq\emptyset \Leftrightarrow M_{b,\max}\in{\cal F}.$
\end{center}
\end{proposition}
\begin{proof}
($\Leftarrow$) This part holds since $M_{b,\max}\in R_I(M_b)\cap{\cal F}$.

($\Rightarrow$)
Suppose that $M_{b,\max}\notin{\cal F}$, i.e., $\mathbf{w}^T\cdot M_{b,\max}>k$.
Notice that $M_{b,\max} = M_b + C_I\cdot \mathbf{y}_{max}$ where $\mathbf{y}_{max}\in \mathbb{N}^{n_I}$ is the maximal IFV at $M_b$.
Since by Eq.~\eqref{eq:pos}, $\mathbf{w}^T\cdot C(\cdot, t)\leq 0$ holds for all $t\in T_I$, we can conclude that for any $M\in R_I(M_b)$ such that $M=M_b+C_I\cdot\mathbf{y}$, $\mathbf{y}\leq\mathbf{y}_{max}$ holds.
Therefore we have:
$$\begin{aligned}
\mathbf{w}^T\cdot M=\mathbf{w}^T\cdot M_b+\mathbf{w}^T\cdot C_I\cdot\mathbf{y}&\geq \mathbf{w}^T\cdot M_b+\mathbf{w}^T\cdot C_I\cdot\mathbf{y}_{max}\\
&=\mathbf{w}^T\cdot M_{b,\max}>k.
\end{aligned}$$
Therefore $R_I(M_b)\cap{\cal F}=\emptyset$.
\end{proof}

Intuitively speaking, since the firing of $T_I$ does not increase (and possibly decreases) the token count of $(\textbf{w},{\textit k})$, the token count at any marking in $R_I(M_b)$ is not less than that of $M_{b,\max}$ (which is reached by firing the maximal number of $T_I$ transitions from $M_b$).
Hence, if the token count of $(\textbf{w},{\textit k})$ at $M_{b, \max}$ exceeds $k$, then the token count at any other marking in $R_I(M_b)$ also exceeds $k$.
%

Finally, we are ready to present the main result of this work.
The following theorem provides a necessary and sufficient condition for the non-blockingness of a net.

\begin{theorem}\label{Theorem1}
Given a plant $G=(N, M_0, {\cal F})$ with $N=(P, T, Pre, Post)$ and $\cal{F} = \mathcal{L}_{(\textbf{w},{\textit k})}$, let $\cal{B}$ be its CI-BRG with respect to $\pi = (T_E,T_I)$.
System $G$ is non-blocking if and only if for any basis marking $M_b$ in $\cal{B}$, there exists a basis marking $M_b^{\prime}$ accessible from $M_b$ and $R_I(M_b^{\prime})\cap{\cal F}\neq\emptyset$.
\end{theorem}
\begin{proof}
(only if) Suppose that $G$ is non-blocking.
For any basis marking $M_b$ in $\cal{B}$, there exists a sequence $\sigma$ such that $M_b[\sigma\rangle M\in{\cal F}$.
We write $\sigma=\sigma_1t_{i_1}\cdots\sigma_nt_{i_n}\sigma_{n+1}$ where all $\sigma_i\in T_I^*, t_{i_j}\in T_E, j=1,\ldots, n$.
Following the procedure in the proof of Theorem 3.8 in \cite{c8}, we can repeatedly move transitions in each $\sigma_j$ ($j\in\{1, \ldots, n\}$) to somewhere after $t_{i_j}$ to obtain a new sequence $\sigma_{min,1}t_{i_1}\sigma_{min,2}t_{i_2}\cdots \sigma_{min, n}t_{i_n}\sigma_{n+1}'$ such that
$$M_b[\sigma_{min,1}t_{i_1}\rangle M_{b,1}[\sigma_{min,2}t_{i_2}\rangle\cdots[\sigma_{min, n}t_{i_n}\rangle M_{b, n}[\sigma_{n+1}'\rangle M$$
where each $\sigma_{\min,j}$ is a minimal explanation of $t_{i_j}$ at $M_{b,j}$ for $j=1,\ldots n$.
Hence, basis marking $M_b'$ is accessible from $M_b$, and $M\in R_I(M_b')$ holds.

(if) Let $M_{b,0}$ be an arbitrary basis marking and $M$ be an arbitrary marking in $R_I(M_{b,0})$. Suppose that in $\cal{B}$ there exists a basis marking $M_{b,n}$ accessible from $M_{b,0}$, i.e.,
\begin{center}
  $M_{b,0}\xrightarrow{(t_1, \mathbf{y}_{min,1})} M_{b,1}\xrightarrow{(t_2, \mathbf{y}_{min,2})} M_{b,2}\cdots\xrightarrow{(t_n, \mathbf{y}_{min,n})} M_{b,n},$
\end{center}
where $R_I(M_{b,n})\cap{\cal F}\neq\emptyset$.
By Proposition \ref{PropThree}, $M_{i,\max}[t_i\rangle$ holds where $M_{i, \max}$ is the i-maximal marking at $M_{b, i}\ (i\in \{1, 2, \ldots, n\})$. Now we prove that from $M$ there exists a firing sequence that reaches $M_{n,\max}$. We use $M\rightarrow M'$ to denote that there exists sequence $\sigma$ such that $M[\sigma\rangle M'$.

By Proposition \ref{PropThree}, it holds that $M\rightarrow M_{0,\max}$. By Proposition \ref{PropOne}, $M_{0,\max}[t_1\rangle M_1$ where $M_1\in R_I(M_{b,1})$. By regarding $M_{b,1}$ and $M_1$ as the original $M_{b,0}$ and $M$, respectively, the above reasoning can be repeatedly applied. Hence, the following trajectory is feasible:
\begin{center}
  $M\rightarrow M_{0, \max}\rightarrow M_1\rightarrow M_{1,\max}\rightarrow \cdots\rightarrow M_n\rightarrow M_{n,\max}$
\end{center}
where $M_i\in R_I(M_{b, i})\ (i\in \{1, 2, \ldots, n\})$. By Proposition \ref{ProFour}, $M_{n,\max}\in{\cal F}$ holds. Therefore, $G$ is non-blocking.
\end{proof}

Theorem~\ref{Theorem1} indicates that the non-blockingness of a plant $G$ can be verified by checking if all basis markings in the CI-BRG are accessible to some basis markings whose implicit reach contains final markings.
By Proposition~\ref{ProFour}, to check $R_I(M_b)\cap{\cal F}\neq\emptyset$ it suffices to test if the i-maximal marking $M_{b, \max}$ is final.
This can be done by solving the following \textit{integer linear programming problem} (ILPP) for all $M_b$ in $\cal{B}$:
    \begin{equation}\label{EqILLP}
\left\{
             \begin{array}{lr}
             \max \quad  \textbf{1}^T\cdot \textbf{y}_I\\

            s.t. \quad\ \ M_b+C_I\cdot \textbf{y}_I=M_{b,max}\\

            \quad \quad \quad \textbf{w}^T\cdot M_{b,max} \leq k\\

            \quad \quad \quad M_{b,max}\in \mathbb{N}^{m}\\

            \quad \quad \quad \textbf{y}_I\in \mathbb{N}^{n_I}
\end{array}\right.
\end{equation}

\begin{discussion}
For simplicity, up to now, we assume a single GMEC characterization ($\mathcal{L}_{(\textbf{w},{\textit k})}$) for the final marking set $\mathcal{F}$; however, our approach can be further generalized to:
   \begin{itemize}
     \item [(a)] $\mathcal{F}$ defined by the conjunction of a finite number of $r$ GMECs (namely an \textit{AND-GMEC}), i.e., $\mathcal{L}_{AND} = \{M\in \mathbb{N}^{m}\mid \textbf{W}^T \cdot M \leq \textbf{k}\} = \bigcap_{(\textbf{w}_i, k_i)\in (\textbf{W},\textbf{k})}\mathcal{L}_{(\textbf{w}_i, k_i)}$, where $\textbf{W} = [\textbf{w}_1\ \textbf{w}_2 \cdots \textbf{w}_r] \in \mathbb{Z}^{m\times r}$ and $\textbf{k} = [k_1\ k_2 \cdots k_r]^T \in \mathbb{N}^{r}$. $(\textbf{w}_i, k_i)\in (\textbf{W}, \textbf{k})$ implies that $(\textbf{w}_i, k_i) = (\textbf{W}(i, \cdot), \textbf{k}(i, \cdot))$ and $i\in \{1, 2, \cdots, r\}$. In such a case, the problem can be solved in the same way by revising the constraint $\textbf{w}^T\cdot M\leq k$ in ILPP (\ref{EqILLP}) as $\textbf{W}^T\cdot M\leq \textbf{k}$;
     \item [(b)] $\mathcal{F}$ defined by the union of a finite number of $s$ GMECs (namely an \textit{OR-GMEC}), i.e., $\mathcal{L}_{OR} = \bigcup_{i\in\{1, 2, \ldots, s\}}\mathcal{L}_{(\textbf{w}_i,k_i)}$ where $\textbf{w}_i\in \mathbb{N}^m$ and $k_i\in \mathbb{N}$. Then, the solution can be done by revising the constraint $\textbf{w}^T\cdot M\leq k$ in ILPP (\ref{EqILLP}) as a disjunctive form of constraints, which can be transformed into its equivalent conjunctive form \cite{cabasino2007identification}.
   \end{itemize}

Note that in both two cases, to construct a CI-BRG, the set $T_E$ needs to be expanded based on any single GMEC to ensure the corresponding set $T_I$ being non-increasing. So, in general, the set $T_E$ is not minimal, which leads to the corresponding CI-BRG larger than the minimal ones. However, in a system in practice, the final markings are usually the states at which all resources (parts/products, vehicles, robots, etc.) are in their ``idle'' or ``finished'' places. The single GMEC (or multiple GMECs) that defines a final set only restricts the tokens in these places. Hence, in practice, the set $T_E$ may not be too much larger than a minimal one; thus, the CI-BRG should not be too large in size and can still be efficiently computed.$\hfill\square$
\end{discussion}

\subsection{Non-blockingness Verification in Non-deadlock-free Nets}

As we have mentioned in the beginning of the previous subsection, Theorem \ref{Theorem1} does not require $G$ be deadlock-free.
In this subsection we discuss the reason behind it.

In this subsection, we discuss the reason behind it.
The following result shows that if $G$ is not deadlock-free, all dead markings are exactly the i-maximal markings of some basis markings in the CI-BRG.
We denote by $\cal D$ the set of dead markings in $R(N, M_0)$, i.e., ${\cal D}=\{M\in R(N, M_0)\mid(\forall t\in T)\, \neg M[t\rangle\}$.

\begin{proposition}\label{prop:deadlock}
Given a plant $G=(N, M_0, {\cal F})$ with $\cal{F} = \mathcal{L}_{(\textbf{w},{\textit k})}$, let $\cal{B}$ be its CI-BRG with respect to $\pi= (T_E, T_I)$.
For any basis marking $M_b$ in $\cal{B}$ such that $R_I(M_b)\cap{\cal D}\neq\emptyset$, $R_I(M_b)\cap{\cal D}=\{M_{b, \max}\}$ holds.
\end{proposition}
\begin{proof}
The ``if'' trivially holds.
For the ``only if'' part, suppose that $R_I(M_b)\cap{\cal D}\neq\emptyset$.
By Proposition~\ref{PropTwo}, all markings $M\in R_I(M_b)$ are coreachable to $M_{b,\max}$, which indicates that all $M\in R_I(M_b)\setminus\{M_{b,\max}\}$ are not dead.
Therefore, the only dead marking in $R_I(M_b)\cap{\cal D}\neq\emptyset$ is $M_{b,\max}$.
\end{proof}

Notice that $R(N, M_0)=\bigcap_{M_b\in{\cal B}}R_I(M_b)$.
Proposition~\ref{prop:deadlock} indicates that all dead markings in $R(N, M_0)$ are $M_{b, \max}$.
Note that we do not need to explicitly compute all the i-maximal markings thanks to the following theorem.

\begin{proposition}\label{prop:deadEnd}
Given a plant $G=(N, M_0, {\cal F})$ with $\cal{F} = \mathcal{L}_{(\textbf{w},{\textit k})}$, let $\cal{B}$ be its CI-BRG with respect to $\pi= (T_E, T_I)$.
For any basis marking $M_b$ in $\cal{B}$, $R_I(M_b)\cap{\cal D}\neq\emptyset$ if and only if $M_b$ does not have any outbound arc in $\cal{B}$.
\end{proposition}
\begin{proof}
(only if) By contrapositive.
Suppose that $M_b$ has an outbound arc labeled by $(t, \mathbf{y})$.
It indicates that there exists a sequence $\sigma$ whose firing vector is $\varphi(\sigma)=\mathbf{y}$ such that $M_b[\sigma\rangle M [t\rangle$.
By Propositions~\ref{PropOne} and \ref{PropTwo}, $M_b[\sigma\rangle M[\sigma'\rangle M_{b, \max} [t\rangle$ holds, which means that $t$ is enabled at the i-maximal marking $M_{b, \max}$, i.e., $M_{b, \max}$ is not dead.
By Proposition~\ref{prop:deadlock}, $R_I(M_b)\cap{\cal D}=\emptyset$ holds.

(if) Suppose that $M_b$ does not have any outbound arc.
This implies that from $M_b$ no explicit transition can fire any more.
Since the $T_I$-induced subnet is acyclic, the number of implicit transitions firable from $M_b$ is bounded, which implies that $M_{b, max}\in R_I(M_b)$ is dead.
\end{proof}
\begin{corollary}\label{cor:deadEnd}
In a CI-BRG, if the i-maximal marking $M_{b, \max}$ of a basis marking $M_b$ is dead and not final, then $M_b$ is not accessible to any $M_b'$ such that $R_I(M_b')\cap{\cal F}\neq\emptyset$.
\end{corollary}
\begin{proof}
This corollary holds since $M_b$ has no outbound arc (Proposition~\ref{prop:deadEnd}) and $R_I(M_b)\cap{\cal F}=\emptyset$ (Proposition~\ref{ProFour}).
\end{proof}

One can see that the case in Corollary~\ref{cor:deadEnd} is included in Theorem~\ref{Theorem1}.
Thus, the non-blockingness can be verified by Theorem~\ref{Theorem1} regardless of the deadlock-freeness of $G$.

\subsection{Algorithm}

Based on the results we have obtained so far, in this subsection we develop a method to verify non-blockingness of a plant using CI-BRG.
\begin{table*}
\caption{Analysis of the reachability graph, minimax-BRG, and CI-BRG for the plant in Fig. \ref{FigBenchmark}.}\label{tablenew1}
\scalebox{0.8}{
\begin{threeparttable}
\begin{tabular}{c|c|c||c|c||c|c||c|c||c||c||c}
\toprule[1pt]
  Run & $\alpha$ & $\beta$ & $|R(N, M_0)|$ & Time\ (s) & $|\mathcal{M_{B_M}}|$ & Time\ (s) & $|\mathcal{M_{B}}|$ & Time\ (s) & Non-blocking? & $|\mathcal{M_{B}}|/|R(N, M_0)|$ & Time ratio \\
  \hline
  1 & 1 & 1 & 1966 & 10 & 284 & 2 & 604 & 1.7 & Yes & 30.7$\%$ & 17$\%$\\
  \hline
  2 & 1 & 2 & 12577 & 277 & 1341 & 15 & 2145 & 11 & Yes & 17$\%$ & 4$\%$\\
  \hline
  3 & 2 & 2 & 76808 & 12378 & 5961 & 179 & 7718 & 105 & No & 10$\%$ & 0.8$\%$\\
  \hline
  4 & 2 & 3 & - & o.t. & 14990 & 1028 & 16438 & 470 & No & - & -\\
  \hline
  5 & 2 & 4 & - & o.t. & 26716 & 3126 & 26648 & 1248 & Yes & - & -\\
  \hline
  6 & 3 & 3 & - & o.t. & 38551 & 6697 & 37118 & 2492 & Yes & - & -\\
  \hline
  7 & 3 & 4 & - & o.t. & 67728 & 22018 & 59315 & 6449 & No & - & - \\
  \hline
  8 & 4 & 4 & - & o.t. & - & o.t. & 101420 & 19491 & No & - & - \\
  \bottomrule[1pt]
\end{tabular}
\begin{tablenotes}
        \item[*] The computing time is denoted by \textit{overtime} (o.t.) if the program does not terminate within 36,000 seconds (10 hours).
      \end{tablenotes}
\end{threeparttable}}
\end{table*}

\begin{algorithm}
\caption{Non-blockingness Verification Using CI-BRG} %
\begin{algorithmic}[1]
\REQUIRE A bounded plant $G=(N, M_0, {\cal F})$
\ENSURE ``$G$ is nonblocking'' $\slash$ ``$G$ is blocking''
\STATE Find a basis partition $\pi = (T_E, T_I)$ where $T_I$ is non-conflicting and non-increasing;
\STATE Construct the CI-BRG $\mathcal{B} = (\mathcal{M_B}, {\rm Tr}, \Delta, M_0)$ of $G$;
\STATE $\hat{\mathcal{M_B}}: = \emptyset$;
\FORALL {$M_b\in \mathcal{M_B}$,}
\IF {{ILPP (\ref{EqILLP})} has a feasible solution,}
\STATE {$\hat{\mathcal{M_B}}: = \hat{\mathcal{M_B}}\cup\{M_b\}$;}
\ENDIF
\ENDFOR
\FORALL {$M_b^{\prime}\in \mathcal{M_B}\setminus\hat{\mathcal{M_B}}$,}
\IF {$\nexists \hat{M_b}\in \hat{\mathcal{M_B}}$, $\nexists \sigma\in {\rm Tr}^*$ s.t. $M_b^{\prime}\xrightarrow{\sigma}\hat{M_b}$,}
\STATE Output ``$G$ is blocking'' and Exit;
\ELSE
\STATE {Continue;}
\ENDIF
\ENDFOR
\STATE {Output ``$G$ is non-blocking'' and Exit.}
\end{algorithmic}\label{AlgoNB}
\end{algorithm}

In brief, Algorithm \ref{AlgoNB} consists of two stages:
\begin{itemize}
  \item Stage (i), steps 1--8: construct the CI-BRG and determine for each basis marking $M_b$ if $R_I(M_b)\cap{\cal F}\neq\emptyset$. The latter done by solving ILPP~(\ref{EqILLP}) for all basis markings;
  \item Stage (ii), steps 9--16: check if any basis marking in $\cal{B}$ is co-reachable to at least a basis marking that is coreachable to some final markings, which can be done by applying a search algorithm (e.g., \textit{Dijkstra}) in the underlying digraph of the CI-BRG, whose complexity is polynomial in the size of $\cal{B}$.
\end{itemize}

\begin{proposition}
Algorithm~\ref{AlgoNB} is correct.
\end{proposition}
\begin{proof}
The set $\hat{\mathcal{M_B}}$ in Algorithm~\ref{AlgoNB} records all basis markings whose i-maximal marking is final.
By Theorem~\ref{Theorem1}, the net is blocking if and only if there exists a basis marking inaccessible to any basis marking in $\hat{\mathcal{M_B}}$.
This coincide with Algorithm~\ref{AlgoNB} who outputs BLOCKING if and only if such a basis marking is detected in Step~10.
\end{proof}


\begin{example}\label{Example3}[Ex.~\ref{Example22} cont.]
  Consider again the plant $G = (N, M_0, \cal{F})$ with $\cal{F} = \mathcal{L}_{(\textbf{w},{\textit k})}$, $\textbf{w} = [0\ 0\ 0\ 1\ 1\ 1]^{T}$, $k = 0$ which is depicted in Fig. \ref{plant2}.
  Its CI-BRG $\mathcal{B}$ with respect to $T_E=\{t_3, t_4, t_6\}$ is shown in Fig. \ref{CPBRG1}.
  Now we execute Algorithm~1 to verify if $G$ is non-blocking.

  First, by solving ILPP~(\ref{EqILLP}) for all $M_b$ in $\mathcal{B}$, we conclude that $\hat{\mathcal{M_B}} = \{M_{b0}, M_{b1}, M_{b2}, M_{b4}, M_{b5}\}$.
  Then, by analyzing the CI-BRG $\mathcal{B}$, it can be inferred that $M_{b3}$ is not co-reachable to any of the basis marking in $\hat{\mathcal{M_B}}$; thus, the system is blocking.$\hfill\diamondsuit$
\end{example}

\section{Simulation Results}\label{Section5}

\subsection{On Efficiency}
We use the parameterized Petri net in Fig. \ref{FigBenchmark} (slightly modified version of the Petri net in Fig.~5 in \cite{cabasino2011discrete}) to test the efficiency of our approach.
All tests are carried out on a PC with Intel Core i7-7700 CPU 3.60 GHz processor and 8.00 GB RAM.
This system consists of $46$ places and $39$ transitions, where the initial marking $M_0$ is parameterized as:
  $M_0 = \alpha p_1+\beta p_{16} + p_{31}+p_{32}+p_{33}+p_{34}+p_{35}+p_{37}+p_{38}+p_{39}+8p_{40}+p_{41}.$
Let $\mathcal{F}=\mathcal{L}_{(\textbf{w},k)}=\{M\in\mathbb{N}^m| \textbf{w}^{T}\cdot M \leq k\}$, where
\begin{equation}\nonumber
    \begin{aligned}
        \textbf{w} = [0\ 1\ 0\ &0\ 0\ 0\ 0\ 1\ 0\ 0\ 0\ 0\ 0\ 0\ 0\ 0\ 1\ 1\ 0\ 0\ 0\ 0\ 1\ 0\ 1\ 0\\
          & 0\ 0\ 0\ 0\ 0\ 0\ 0\ 0\ 0\ 1\ 0\ 0\ 0\ 0\ 0\ 0\ 0\ 0\ 0\ 0]^{T}
    \end{aligned}
\end{equation}
and $k=4$ (for runs 1$\--$4) or $k = 7$ (for runs 5$\--$8) to test non-blockingness of this plant for all cases.
\begin{figure}[t]
  \centering
  \includegraphics[width=11cm]{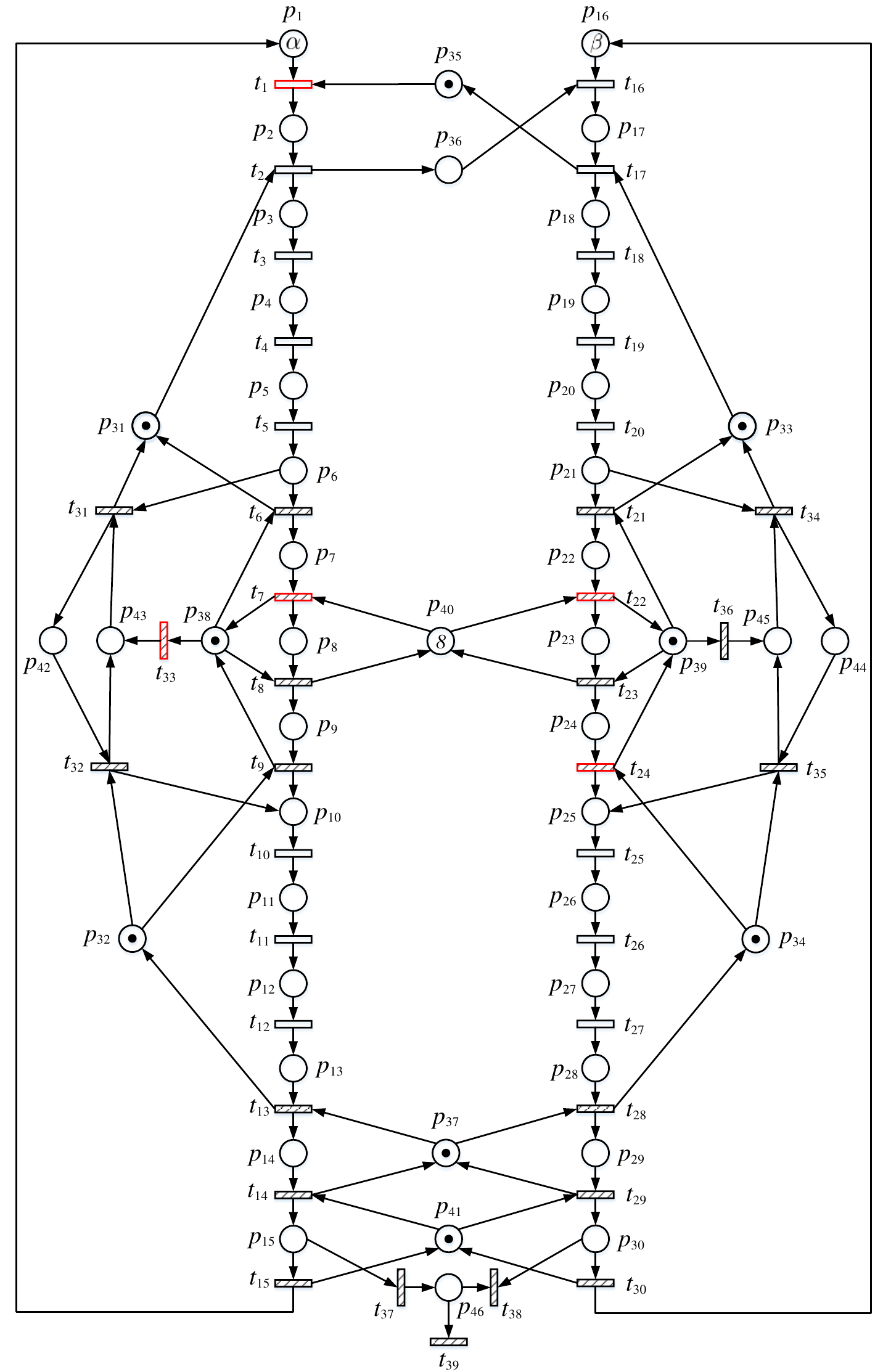}\\
  \caption{A parameterized plant $G = (N, M_0, \mathcal{F})$.}\label{FigBenchmark}
\end{figure}

Firstly, it can be inferred that $T_{conf} = \{t_6, t_7, t_8, t_9, t_{13}, t_{14}, t_{15}, t_{21},t_{22},t_{23},t_{24},t_{28},t_{29},t_{30},t_{31},t_{32},t_{33},
t_{34},\\
t_{35},t_{36},t_{37},t_{38},t_{39}\}$ (marked in shadow) and $T_{inc} = \{t_1, t_7, t_{22}, t_{24}, t_{35}\}$ (boxed in red). Since the subnet induced by all transitions $t\in T\setminus (T_{conf}\cup T_{inc})$ is acyclic, a non-conflicting and non-increasing $T_I$ can be obtained, i.e., $T_I = T\setminus (T_{conf}\cup T_{inc})$ and therefore $T_E = T\setminus T_I$. Thus, the corresponding CI-BRG $\mathcal{B}$ can be constructed based on $\pi = (T_E, T_I)$.
Moreover, for comparison, we compute the minimax-BRG \cite {gu2020verification} of the plant (denoted as $\mathcal{B_M}$) in Fig. \ref{FigBenchmark} by considering another basis partition $\pi^{\prime} = (T_E^{\prime}, T_I^{\prime})$ where $T_E^{\prime} = \{t_1, t_6, t_7, t_9, t_{14}, t_{21}, t_{23}, t_{24}, t_{29}, t_{32}, t_{34}, t_{35}\}$ and $T_I^{\prime} = T\setminus T_E^{\prime}$. The set of all minimax basis markings in $\mathcal{B_M}$ is denoted as $\mathcal{M_{B_M}}$.

In Table \ref{tablenew1}, for different values of $\alpha$ and $\beta$, the number of basis markings in $\mathcal{M_{B}}$ in the CI-BRG $\mathcal{B}$, minimax basis markings in $\mathcal{M_{B_M}}$ in the minimax-BRG $\mathcal{B_M}$, and all reachable markings $M\in R(N, M_0)$, as well as their computing times are listed in columns 1$\--$9.
Meanwhile, the non-blockingness of each case and the ratios of node number and time assumptions in terms of reachability graph and $\mathcal{B}$ are respectively reported in columns 10$\--$12.
Through the results, we conclude that the CI-BRG-based approach outperforms that of the RG-based method in this plant for all cases; whereas, with the expansion of the system scale, CI-BRG shows the potential to be more efficient than the minimax-BRG in this study.

\subsection{On Practice}
To echo the motivation of our research and illustrate the verification process of our approach in detail, in this subsection, we test a real-world example, which refers to a \textit{Hospital Emergency Service System} \cite{li2014robust} that is modelled by Petri nets. All tests are carried out on a PC with an Intel Core i7-7700 CPU 3.60 GHz processor and 8.00 GB RAM.

    This example demonstrates the medical service process \cite{sampath2008control} from patients’ arrival to departure after their treatment as modelled by a plant $G = (N, M_0, \mathcal{F})$ in Fig. \ref{RealWorldPN}.
    Physically, place $p_1$ represents the count of patients in the emergency department.
    The Petri net $N$ in $G$ consists of $22$ places and $22$ transitions.
    Consider the initial marking $M_0 = 4 p_1+ 4 p_{11} + 4  p_{18}+4 p_{19}.$
    \begin{figure}[]
  \begin{center}
\includegraphics[width=5.5cm]{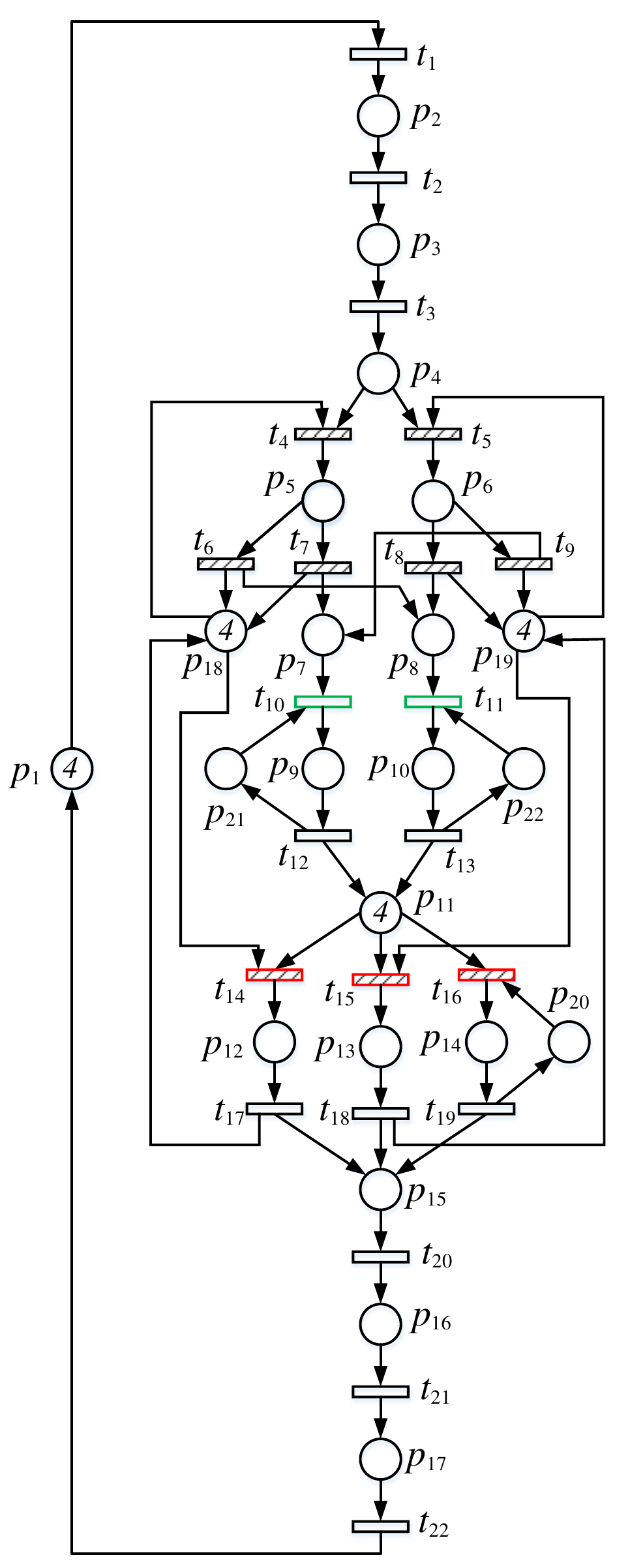}
\caption{The Hospital Emergency Service System modelled by a plant $G = (N, M_0, \mathcal{F})$.}\label{RealWorldPN}
\end{center}
\end{figure}
Let $\mathcal{F}=\mathcal{L}_{(\textbf{w},k)}=\{M\in\mathbb{N}^m| \textbf{w}^{T}\cdot M \leq k\}$, where
\begin{equation}\nonumber
    \begin{aligned}
        \textbf{w} = [1\ 1\ 1\ 1\ 1\ 1\ 1\ 1\ 1\ 1\ 0\ 1\ 1\ 1\ 1\ 1\ 1\ 0\ 0\ 0\ 1\ 1]^T,
    \end{aligned}
\end{equation}
$k= 6$ for Case I, and $k= 8$ for Case II.
Next, based on Algorithm 1, we use our approach (CI-BRG-based one) to test non-blockingness of this plant $G = (N, M_0, \mathcal{F})$ in two cases.

Since the construction of $\cal{B}$ is irrelevant to the parameter $k$, we first find a basis partition $\pi = (T_E,T_I)$ where $T_I$ is non-conflicting and non-increasing, and then construct the corresponding CI-BRG.
Based on plant $G$, it can be inferred that $T_{conf} = \{t_4, t_5, t_6, t_7, t_{8}, t_{9}, t_{14}, t_{15},t_{16}\}$ (marked in shadow) and $T_{inc} = \{t_{14}, t_{15}, t_{16}\}$ (boxed in red).
To ensure the acyclicity of the $T_I$-induced subnet, let transitions $t_{10}$ and $t_{11}$ be explicit (boxed in green).
Since the sub-net induced by all transitions $t\in T\setminus (T_{conf}\cup T_{inc}\cup\{t_{10}, t_{11}\})$ is acyclic, a non-conflicting and non-increasing $T_I$ can be obtained, i.e., $T_I = T\setminus (T_{conf}\cup T_{inc}\cup\{t_{10}, t_{11}\})$ and therefore $T_E = T\setminus T_I$. Thus, the corresponding CI-BRG $\mathcal{B}= (\mathcal{M_B}, {\rm Tr}, \Delta, M_0)$ can be constructed based on $\pi = (T_E, T_I)$.
The CI-BRG $\mathcal{B}$ contains 3863 nodes and can be constructed in 33 seconds (in contrast, the RG cannot be constructed within 36,000 seconds).
Here, we do not show the graphical representation of $\mathcal{B}$ due to its size.
Next, we illustrate the verification process for Cases I and II respectively as follows.
\begin{itemize}
                    \item Case I ($k= 6$):

                   \item [] For all $M_b\in \mathcal{M_B}$, we add $M_b$ into the set $\hat{\mathcal{M_B}}$ if $R_I(M_b)\cap{\cal F}\neq\emptyset$.
By solving ILPP (\ref{EqILLP}), the set $\hat{\mathcal{M_B}}$ can be obtained and $|\hat{\mathcal{M_B}}| = 818$.

Next, we proceed to stage (ii). According to computation, it is concluded that there exists basis marking in $\cal{B}$ that is not co-reachable to any  basis marking in $\hat{\mathcal{M_B}}$; thus, the plant $G$ is inferred to be blocking.

                    \item Case II ($k= 8$):

                    \item [] Similarly, by solving ILPP (\ref{EqILLP}), the set $\hat{\mathcal{M_B}}$ can be obtained and $|\hat{\mathcal{M_B}}| = 3863$.

Next, we proceed to stage (ii). It is concluded that any basis marking in $\cal{B}$ is co-reachable to at least a basis marking in $\hat{\mathcal{M_B}}$; thus, the plant $G$ is non-blocking.
                  \end{itemize}

\section{Conclusion}
We have developed a novel method for non-blockingness verification in Petri nets.
By adopting a basis partition with transition set $T_I$ being non-conflicting and non-increasing, we have proposed a particular type of BRGs called the CI-BRGs.
Based on CI-BRGs, we have proposed a necessary and sufficient condition for non-blocking verification proved that a net is non-blockingness.
Our method can be applied to both deadlock-free nets and non-deadlock-free ones.
Simulation shows that the approach we have proposed achieves practical efficiency.
In the future, we plan to tackle the non-blockingness enforcement problem by applying BRGs.

\bibliographystyle{plain}        
\bibliography{sample_2021cdc}

\end{document}